\documentclass[letterpaper, 10 pt, conference]{ieeeconf}

\IEEEoverridecommandlockouts                              
\usepackage{subfigure}

\usepackage{amsthm}

\usepackage{amssymb, amsmath, amsfonts}
\usepackage{tikz}
\usepackage{url}
\usetikzlibrary{patterns}
\usetikzlibrary{external}
\usepackage{pgfplots}
\usepackage{graphicx}   
\usepackage[latin1]{inputenc} 
\usepackage{float}
\usepackage{graphicx}
\usepackage{comment}

\theoremstyle{plain}
\newtheorem{theorem}{Theorem}
\newtheorem{lemma}{Lemma}

\newtheorem{corollary}{Corollary}

\theoremstyle{definition}
\newtheorem{proposition}{Proposition}
\newtheorem{definition}{Definition}
\newtheorem{remark}{Remark}

\newtheorem{problem}{Problem}

\newlength\figureheight
\newlength\figurewidth

\ifnum\pdfshellescape=1
\fi

\newif\ifcomment
\commenttrue
\commentfalse
\ifcomment
\newcommand{\comments}[1]{\textcolor{blue}{#1}}
\newcommand{\commentsr}[1]{\textcolor{red}{#1}}
\else
\newcommand{\commentsr}[1]{}
\newcommand{\comments}[1]{}
\fi

\newif\ifshorten
\shortentrue
\shortenfalse
\ifshorten
\newcommand{\shorten}[1]{}
\else
\newcommand{\shorten}[1]{#1}
\fi

\newif\iffinalcdc
\finalcdctrue
\finalcdcfalse
\iffinalcdc
\newcommand{\finalcdc}[1]{}

\renewcommand{\baselinestretch}{0.975}

\else
\newcommand{\finalcdc}[1]{#1}
\fi

\newcommand{\T}[1]{T_{\mathcal S^{(#1)}}}
\newcommand{\Dp}[1]{\mathbb D_p^{(#1)}}
\newcommand{\Do}[1]{\mathbb D_o^{(#1)}}
\newcommand{\invT}[1]{T^{-1}_{\mathcal S^{(#1)}}}

\begin{document}

\title{\LARGE \bf An Iterative Abstraction Algorithm for Reactive
Correct-by-Construction Controller Synthesis}

\author{
Robert Mattila, Yilin Mo and Richard M. Murray%
\thanks{This work is supported in part by IBM and UTC through the Industrial
Cyber-Physical Systems Center (iCyPhy) consortium.}%
\thanks{Robert Mattila is with the Department of Automatic Control, KTH Royal
Institute of Technology, Stockholm, Sweden. {\tt\small rmattila@kth.se}}%
\thanks{Yilin Mo and Richard M. Murray are with the Department of Control and
    Dynamical Systems, California Institute of Technology, Pasadena, CA, USA.
{\tt\small yilinmo@caltech.edu}, {\tt\small murray@cds.caltech.edu}}
}

\maketitle
\thispagestyle{empty}
\pagestyle{empty}

\begin{abstract} In this paper, we consider the problem of synthesizing
    correct-by-construction controllers for discrete-time dynamical systems.
    A commonly adopted approach in the literature is to abstract the dynamical
    system into a \emph{Finite Transition System} (FTS) and thus convert the problem
    into a two player game between the environment and the system on the FTS.
    The controller design problem can then be solved using synthesis tools for
    general linear temporal logic or generalized reactivity(1) specifications.
    In this article, we propose a new abstraction algorithm. Instead of
    generating a single FTS to represent the system, we generate two FTSs,
    which are under- and over-approximations of the original dynamical system.
    We further develop an iterative abstraction scheme by exploiting the
    concept of winning sets, i.e., the sets of states for which there exists
    a winning strategy for the system.  Finally, the efficiency of the new
    abstraction algorithm is illustrated by numerical examples.
\end{abstract}

\section{Introduction}

The systems that are considered for control purposes have changed fundamentally
over the last few decades. Driven by the advancements in computation and
communication technologies, the systems of today are highly complicated with
large amounts of components and interactions, which poses great challenges
to controller design. This is exemplified in \cite{darpa-car} where the
controller for an autonomous vehicle became so unwieldy that it was impossible
to foresee the failure of it, resulting in a crash.

In order to tame the complexity of modern control systems, synthesis of
correct-by-construction control logic based on temporal logic specifications
has gained considerable attention in the past few years. A commonly
adopted approach is to construct a \emph{Finite Transition System} (FTS) which
serves as a symbolic model of the original control system, which typically has
infinitely many states. The controller, which is represented by a finite state
machine, can then be synthesized to guarantee certain specifications on the
system by leveraging formal synthesis tools~\cite{Piterman2006}. Such a design
procedure has been applied to various fields including
robotics (e.g. \cite{Karaman09,KloetzerB10,Bhatia2010,KGF,Fainekos06}), autonomous
vehicle control~\cite{rhtlp}, smart-buildings~\cite{Raman14} and aircraft power
system design~\cite{Nuzzo13}.

One of the main challenges of this approach is in the abstraction of the
control system, whose state space is continuous and potentially high
dimensional, into a finite state model. Zamani et al.~\cite{zamani2012symbolic}
propose an abstraction algorithm based on approximate simulation relations and
alternating approximate simulation relations. They prove that if certain
continuity assumptions on the system trajectory hold, then an FTS can be
generated by partitioning the state space into small hypercubes. Similar ideas
are also presented in \cite{tabuada:approx} and \cite{tabuada:book}.

A different, iterative, approach has been proposed that first generates
a coarse model of the original system and then refines the model based on
reachability computations \cite{rhtlp, rhtlpfds}. This algorithm has been
implemented in a software package, namely \texttt{TuLiP}~\cite{tulip}, and will
be compared to the method proposed in this paper.

Most of the algorithms available in the literature generate the finite state
model independently of the system specifications. As such, the abstracted model
can be used for any possible specification. However, this typically leads to
a partition of the state space into equally fine regions everywhere. As
a consequence, the time complexity of such general abstraction procedures is
quite high and it increases with the dimension of the system.

In this article, in hope to reduce the computational complexity of the
abstraction algorithm, we create the finite state models of the system by
exploiting the structure of the specifications.  To be specific, we create two
FTS models for the control system, where one is an over-approximation of the
control system and the other is an under-approximation. By solving the
synthesis problem on both FTSs, we can categorize the points in the state space
into, what we refer to as, \emph{winning, losing} and \emph{maybe} sets.  Conceptually, the
winning set contains those points for which a correct controller is
known, i.e., roughly, a controller that can fulfill the given specifications.
On the other hand, the losing set contains those points for which we
know that no correct controller exists.  Lastly, the maybe set
represents the points for which the existence of a correct controller is
not yet known since the current model is not fine enough to represent
the original system. One can view the winning and losing sets as
the ``solved'' regions and the maybe set as the ``unsolved'' region. We
can thus focus our computational power on refining the abstraction of the
regions of the state space that lie in the maybe set, while
leaving the current winning and losing sets intact.

The merits of our proposed algorithm are twofold:
\begin{enumerate}
  \item Instead of partitioning the state space into equally fine regions, we
can concentrate the computational power on the regions for which the existence of
a correct controller is not yet known.
  \item Compared to the abstract algorithm proposed in~\cite{rhtlp, rhtlpfds,
      tulip}, for the case that the specifications are unrealizable (for the
      original continuous control system), our algorithm can provide
      proof that no correct controller exists.
\end{enumerate}

Ideas similar to our proposed method have been presented in
\cite{alfaro:solving} and \cite{moor:learning}. Our algorithm does however
allow us to skip some reachability calculations when performing the
refinement, and can as such be seen as an extension.

\comments{Somehow add references to \cite{camara:synthesis},
\cite{alfaro:solving} and \cite{moor:learning}.}

The rest of the paper is organized as follows: In
Section~\ref{sec:preliminary}, we provide an introduction to transition systems
and linear temporal logic. The problem of abstracting a discrete-time control
system into FTSs is proposed in Section~\ref{sec:problem}.
The abstraction algorithm is then discussed in Section~\ref{sec:main}. Two
numerical examples are provided in Section~\ref{sec:simulation} to illustrate
the effectiveness of the proposed algorithm. Finally,
Section~\ref{sec:conclusion} concludes the paper.

\section{Preliminaries}
\label{sec:preliminary}

\comments{Robert: Something is fishy with our definitions of dom(S), S and
$\mathcal{S}$.}

Most of the definitions in this section can  be found in~\cite{rhtlp}, but are
included in this section for the sake of completeness. For a more thorough
presentation, see e.g.~\cite{pomc}.

\subsection{Transition Systems and Linear Temporal Logic}

\commentsr{Review: Explain the rationale behind a system, i.e., what does it
    represent. You don't seem to introduce the concept of environment, but use
    it several times throughout the paper. Since the presence of an
    environment, and its interaction with the system are quite fundamental in
    the type of synthesis problem you consider, this should be explained and
defined.}

\comments{Robert: Other tulip-related papers dont appear to do this.}

\begin{definition}
  A \emph{system} consists of a set $V$ of variables. The \emph{domain} of $V$, denoted by $dom(V)$, is the set of valuations of $V$. A \emph{state} of the system is an element $v \in dom(V)$.
\end{definition}
In this paper, we consider a system with a set $V = S \cup \mathcal{E}$ of variables. The domain of $V$ is given by $dom(V) = dom(S) \times dom(\mathcal{E})$, where a state $\varsigma \in dom(S) $ is called the \emph{controlled state} and a state $e \in dom(\mathcal{E})$ the uncontrolled \emph{environmental state}. As a result, the state $v$ can be written as $(\varsigma,e)$. We further assume that the set $dom(\mathcal E)$ is finite.

\commentsr{Review:Does the system only consist of the variables, or also of their domains (or do variables already contain their domains then variables should be defined as such)?How are dynamics of the system defined?}

\comments{Robert: This definition is taken from another tulip-paper}

\begin{definition}
  A \emph{transition system} (TS) is a tuple $\mathbb{T} := (\mathcal{V}, \mathcal{V}_{init}, \rightarrow)$ where $\mathcal{V}\subseteq dom( V)$ is a set of states, $\mathcal{V}_{init} \subseteq \mathcal{V}$ is a set of initial states and $\rightarrow \subseteq \mathcal{V} \times \mathcal{V}$ is a transition relation. Given states $\nu_i, \nu_j \in \mathcal{V}$, we write $\nu_i \rightarrow \nu_j$ if there is a transition from $\nu_i$ to $\nu_j$ in $\mathbb{T}$. We say that $\mathbb{T}$ is a \emph{finite transition system} (FTS) if $\mathcal{V}$ is finite.
\end{definition}

\begin{definition}
  An \emph{atomic proposition} is a statement on system variables $\nu$ that has a unique truth value for a given value of $\nu$. Letting $\nu \in dom(V)$ and $p$ be an atomic proposition, we write $\nu \models p$ if $p$ is true at the state $\nu$.
\end{definition}

We will use \emph{Linear Temporal Logic} (LTL), which is an extension of regular propositional logic that introduces additional temporal operators, to formulate specifications on a system. In particular, apart from the standard logical operators negation ($\neg$), disjunction ($\vee$), conjunction ($\wedge$) and implication ($\Rightarrow$), it includes the temporal operators next ($\bigcirc$), always ($\square$), eventually ($\lozenge$) and until ($\mathcal{U}$). LTL formulas are defined inductively as
\begin{enumerate}
  \item Any atomic proposition $p$ is an LTL formula.
  \item Given the LTL formulas $\varphi$ and $\psi$; $\neg \varphi$, $\varphi
\vee \psi$, $\bigcirc \varphi$ and $\varphi \; \mathcal{U} \;  \psi$ are LTL
formulas as well.
\end{enumerate}

\begin{definition}
\commentsr{Review: \emph{execution} has not been defined}
The \emph{satisfaction relation} $\models$ between an execution (infinite
sequence of system states) $\sigma = \nu_0\nu_1\dots$ and an LTL formula is defined inductively as
 \begin{itemize}
   \item $\sigma\models p$ if $\nu_0\models p$.
   \item $\sigma\models \neg \varphi$ if $\sigma$ does not satisfy $\varphi$.
   \item $\sigma\models \varphi\vee \psi$ if $\sigma\models \varphi$ or $\sigma\models \psi$.
   \item $\sigma\models \bigcirc \varphi$ if $\nu_1\nu_2\dots\models \varphi$.
   \item $\sigma\models \varphi\;\mathcal{U}\; \psi$ if there exists an $i\geq 0$, such that $\nu_i\nu_{i+1}\dots\models \psi$ and for any $0\leq k < i$, $\nu_k\nu_{k+1}\dots\models\varphi$.
 \end{itemize}
\end{definition}
\finalcdc{For a more in-depth explanation of LTL, see~\cite{pomc}.}

\iffinalcdc
It is well known that the complexity of synthesizing a controller for a general
LTL formula is double exponential in the length of the given
specification~\cite{Pnueli1989}. However, for a specific class of LTL formulas,
namely those known as \emph{Generalized Reactivity(1)} (GR1) formulas, an
efficient polynomial time algorithm~\cite{Piterman2006} exists. As a result, in
this article, we will restrict the specification $\varphi$ to be a GR1 formula,
which takes the following form: $\varphi = \bigwedge_{i=1}^M \square\lozenge
p_i \Longrightarrow\bigwedge_{j=1}^N \square\lozenge q_j$,
where each $p_i,q_j$ is a Boolean combination of atomic propositions.
\else
It is well known that the complexity of synthesizing a controller for a general
LTL formula is double exponential in the length of the given
specification~\cite{Pnueli1989}. However, for a specific class of LTL formulas,
namely those known as \emph{Generalized Reactivity(1)} (GR1) formulas, an
efficient polynomial time algorithm~\cite{Piterman2006} exists. As a result, in
this article, we will restrict the specification $\varphi$ to be a GR1 formula,
which takes the following form:
\begin{equation}
  \varphi = \bigwedge_{i=1}^M \square\lozenge p_i \Longrightarrow\bigwedge_{j=1}^N \square\lozenge q_j,
  \label{eq:gr1}
\end{equation}
where each $p_i,q_j$ is a Boolean combination of atomic propositions.
\fi

\comments{Robert: Note that I changed $\rightarrow$ to $\Longrightarrow$ in (1), is that correct? What are N and M in (1)?}

\comments{Yilin: I think $\Longrightarrow$ is fine. $N$ and $M$ are the number of assumptions and guarantees respectively.}

\commentsr{Review: Note that the synthesis for GR1 is still exponential in the size of the variable domains, which should probably be highlighted here: as far as I understand your algorithm, the complexity of synthesis is exponential in the number of partitions you generate.  And regarding (1): The definition of GR1 formulae in $[8]$ includes safety parts and initial parts. This encodes your model, but this should be mentioned. Moreover, you don't mention the distinction between environment and system variables here, and that assumptions and guarantees have asymmetric requirements on the variables used. This is due to the turn-based execution assumed in the synthesis algorithm (i.e. environment moves first, then system reacts), and should be highlighted in your text.}

\comments{Robert: Uh..}

\comments{Yilin: I think the complexity of GR1 synthesis is not our concern for this paper. So no need to address that. Regarding the comment on asymmetry, could you check the reference to see if they have any assumptions on the atomic proposition $p_i$ and $q_j$? (for example, $p_i$ consists only env variables and $q_j$ consist only system variables)}

\subsection{Winning Controllers and Winning Sets}

\begin{definition}
  \label{def:controller}
  A \emph{controller} for a transition system $(\mathcal{V},
  \mathcal{V}_\text{init},$ $ \rightarrow)$ and environment $\mathcal{E}$ is an
  ordered set of mappings $\gamma_t:\mathcal{S} \times \mathcal{E}^t
  \rightarrow \mathcal{S}$, i.e.,
    $\gamma \triangleq ( \gamma_1, \gamma_2, \dots, \gamma_t, \dots )$,
  each taking the initial controlled state $\varsigma[0]$ and all the environmental actions up to time $t-1$, $e[0] \dots e[t-1]$, giving another state in $\mathcal{S}$ as output. Furthermore, a controller $\gamma$ is called \emph{consistent} if for all $t$ and $\varsigma[0], \, e[0], \, \dots, \,$ $e[t+1]$, the following transition relation is satisfied:
    $(\gamma_{t}(\varsigma[0], e[0], \dots, e[t-1]),e[t]) 
    \rightarrow (\gamma_{t+1}(\varsigma[0], e[0], \dots, e[t]),e[t+1])$.
\end{definition}
\commentsr{Review:Presumably the index t goes to infinity with time, so how do you represent such controllers? You mention in Remark 1 that finite memory is sufficient, but never say what finite memory means, and how it gives rise to finitely-representable controllers from your Definition 5. Please also introduce bracket notation $[.]$ somewhere.}
\comments{Robert: Skip this?}
\comments{Yilin: Yes. I think our answer is fine since we refer to another paper. If he wants, he can check that paper to see what does finite memory mean. On the other hand, we use subscript to indicate time when we define execution and $[.]$ here. It is better to be consistent and hence we should change the definite of execution to $\nu[0]\nu[1]\dots$. The subscript in the definition of controller is fine. }

\begin{definition}
  Given an infinite sequence of environmental states $e[0]e[1]\dots$, a \emph{controlled execution} $\sigma$ using the controller $\gamma$ and starting at $\varsigma[0]$ is an infinite sequence  
   $\sigma = \nu_0\nu_1\dots =  (\varsigma[0], e[0])(\varsigma[1], e[1])\dots$,
  such that $\varsigma[t+1] = \gamma_t(\varsigma[0], e[0], \dots, e[t+1]).$
  \label{def:contr_exec}
\end{definition}
\begin{definition}
  A set of controlled states $\mathcal W$ is \emph{winning} if there exists
a consistent controller $\gamma$, such that for any infinite sequence of
$e[0]e[1]\dots$ and any initial controlled state $\varsigma[0]\in \mathcal W$,
the controlled execution $\sigma$ using controller $\gamma$ starting at
$\varsigma[0]$ satisfies the GR1-specification $\varphi$. The corresponding controller $\gamma$ is
called a \emph{winning controller} for $\mathcal W$.
  \label{def:win_contr_at}
\end{definition}


The following observations are important for the rest of the paper:
\begin{proposition}
  Let $\{\mathcal W_i\}_{i\in \mathcal I}$ be a collection of winning sets, then the set $\bigcup_{i\in \mathcal I} \mathcal W_i$ is also winning.
  \label{prop:union_of_W}
\end{proposition}
\comments{Robert: I shifted the proof of this to appendix.}
As a result, there exists a largest winning set, which leads to the following definition:
\begin{definition}
  The \emph{largest winning set}, $W$, of a transition system $\mathbb T$, for the
specification $\varphi$, is defined as the union of all winning sets, i.e.,
  \begin{equation}
    W(\mathbb T,\varphi) = \bigcup_{\mathcal W\text{ is winning}}\mathcal W.
  \end{equation}
  The \emph{losing set}, $L$, is defined as 
  \begin{equation}
    L(\mathbb T,\varphi) = dom(S) \; \backslash \; W(\mathbb T,\varphi).
  \end{equation}
 A state $\varsigma$ is called a \emph{losing state} if $\varsigma \in L(\mathbb T,\varphi)$.
  \label{def:largestwin}
\end{definition}
\begin{remark}
  Notice that the controllers defined in Definition~\ref{def:controller} have infinite memory (since they require all environmental actions $e[0]e[1]\dots$).  However, from \cite{Piterman2006}, we know that for a \emph{finite} transition system, if a winning controller exists, there will also exist a winning controller with finite memory.
\end{remark}
\section{Problem Formulation}
\label{sec:problem}
\commentsr{Review: Define \emph{control system} before you use it!}

\comments{Robert: This is our definition of a control system...? Should I add something like: A control system is a dynamical system where the dynamics are not completely autonomous (..)? That just sounds silly.}
\comments{Yilin: You do not need to change the text.}

We consider the following discrete-time control system:
\begin{equation}
  \begin{split}
    s[t+1] &= f(s[t], u[t]), \\
    u[t] &\in U,\, s[t] \in dom(S), \\
    s[0] &\in S_\text{init},
  \end{split}
  \label{eq:system_model}
\end{equation}
where $dom(S)\subseteq \mathbb R^n$, $S_\text{init} \subseteq dom(S)$ is the
set of possible initial states, $U \subseteq \mathbb{R}^m$ is the admissible
control set and $f$ the system dynamics (possibly non-linear). It is evident
that the discrete-time control system is completely characterized by
$f,\,U,\,dom(S)$ and $S_\text{init}$, which leads to the following formal definition:
\begin{definition}
  A \emph{discrete-time control system} $\Sigma$ is a quadruple $\Sigma \triangleq (f,\,U,\,dom(S),\,S_\text{init})$.
\end{definition}

\commentsr{Review: Where does $\mathcal{E}$ get quantified in the definition? It is also not mentioned in (4).}

\comments{Robert: It is not defined more closely in any other tulip (conference) paper. Also, the environment is not relevant for the dynamical model of the system. So skip this comment?}

\comments{Yilin: I think we should say that we assume the ``controlled'' states evolves according to a discrete-time control system. Then maybe give an example, like the one in [19], where the controlled states is the position of the autonomous vehicle and the env variable represent random obstacles. The goal is to avoid collision with the obstacle and visit some places. Therefore, the env is not in the dynamic equation (4) but in the specifications. My description may not be accurate. Please check the example section in [19].}

A discrete-time control system $\Sigma$ can be converted into a transition system in the following manner:
\begin{definition}
  Let $\Sigma \triangleq (f,\,U,\,dom(S),\,S_\text{init})$ be a discrete-time control system. The transition system $TS(\Sigma) = (\mathcal V, \mathcal V_{init},$ $\rightarrow)$ associated with $\Sigma$ is defined as:
  \begin{itemize}
    \item $\mathcal V = dom(S)\times dom(\mathcal E)$.
    \item $\mathcal V_{init} = S_\text{init}\times dom(\mathcal E)$.
    \item For any $(s_1,e_1),(s_2,e_2)\in \mathcal{V}$, $(s_1,e_1) \rightarrow (s_2,e_2)$ if and only if there exists $u\in U$, such that $s_2 = f(s_1,u)$.
  \end{itemize}
\end{definition}

The problem of controller synthesis for the discrete-time control system $\Sigma$ can be written as a controller synthesis problem for $TS(\Sigma)$ as follows:
\begin{problem}
  \emph{Realizability:} Given $TS(\Sigma)$ and a specification $\varphi$, decide whether $S_{init}$ is a winning set.
  \label{problem:realizability}
\end{problem}
\begin{problem}
  \emph{Synthesis:} Given $TS(\Sigma)$ and a specification $\varphi$, if $S_{init}$ is winning, construct the winning controller $\gamma$.
  \label{problem:synthesis}
\end{problem}
In general, Problem~\ref{problem:realizability} and \ref{problem:synthesis} are
very challenging, even for a very simple formula $\varphi$~\cite{VidEtal2000,
RakEtal2006}. As a result, we will attack this problem by leveraging the tools
developed for controller synthesis for FTSs. The main difficulty in directly
applying these techniques is that $TS(\Sigma)$ has infinitely (uncountably)
many states. In the next section, we develop abstraction techniques to convert
$TS(\Sigma)$ into FTSs.

\section{Abstraction Algorithm}
\label{sec:main}
In this section, we abstract $TS(\Sigma)$ into two FTSs with the same set of states by partitioning the state space into equivalence classes. We will refer to $s \in dom(S)$ as a \emph{continuous state} for $TS(\Sigma)$ and any state $\varsigma$ of the FTSs as a \emph{discrete state}. 


\subsection{Constructing the Initial Transition Systems}

Our proposed method builds upon the idea of creating an over-approximation and an under-approximation of the reachability relations of the system. 
To this end, we (iteratively) construct two FTSs. One that we will refer to as the
\emph{pessimistic FTS} and one that we will refer to as the \emph{optimistic
FTS}. We introduce the notation $\mathbb{D}_{o}^{(i)} = (\mathcal{V}^{(i)},
\mathcal{V}_{\text{init}}^{(i)}, \rightarrow_o^{(i)})$ and
$\mathbb{D}_{p}^{(i)} = (\mathcal{V}^{(i)}, \mathcal{V}_{\text{init}}^{(i)},
\rightarrow_p^{(i)})$, respectively, for the $i$th iteration of these FTSs
(i.e. those constructed in the $i$th iteration of the algorithm).

To simplify the notation, we define two reachability relations as:
\begin{definition}
  The relation $\mathcal{R}_p: 2^{dom(S)} \times 2^{dom(S)} \rightarrow \{0, 1\}$ is defined such that $\mathcal{R}_p (X, Y) = 1$ if and only if for all $x\in X$, there exists an $y\in Y$ and $u\in U$, such that $f(x,u) = y$.
\end{definition}
\begin{definition}
  The relation $\mathcal{R}_o: 2^{dom(S)} \times 2^{dom(S)} \rightarrow \{0, 1\}$ is defined such that $\mathcal{R}_o (X, Y) = 1$ if and only if there exist $x \in X$, $y\in Y$ and $u\in U$, such that $f(x,u)=y$.
\end{definition}
\shorten{
\begin{remark}
  Informally, $\mathcal{R}_p$ indicates whether there is some control action for every continuous state in a region $X$ that takes that state to some state in the region $Y$ in one time step. $\mathcal{R}_o$ indicates whether there is some point in $X$ that can be controlled to $Y$ in one time step. The results can be generalized to longer horizon lengths, but for simplicity we only consider reachability in one time step.
\end{remark}
}
We further define a partition function of the continuous state space $dom(S)$:
\begin{definition}
  A \emph{partition function} of $dom(S)$ is a mapping $T_{\mathcal S}:
dom(S)\rightarrow \mathcal S$. The inverse of $T_{\mathcal S}$ is defined as $T^{-1}_{\mathcal S}:\mathcal S\rightarrow 2^{dom(S)}$, such that
  \begin{displaymath}
    T^{-1}_{\mathcal S}(\varsigma) = \{s\in dom(S):T_{\mathcal S}(s) = \varsigma\}.
  \end{displaymath}
\end{definition}
\begin{definition}
  The partition function $T_{\mathcal S}$ on $dom(S)$ is called
\emph{proposition preserving} if for any atomic proposition $p$ and any pair of continuous states $s_a,s_b\in dom(S)$, which satisfy $T_{\mathcal S}(s_a) = T_{\mathcal S}(s_b)$, we have that $s_a\models p$ implies that $s_b\models p$.
\end{definition}

If $T_{\mathcal S}$ is proposition preserving, then we can label the discrete states with atomic propositions. To be specific, we say $\varsigma \models p$ if and only if for every $s\in T^{-1}_{\mathcal S}(\varsigma)$, we have that $s\models p$. 

To initialize the abstraction algorithm, we assume that we are given the atomic
propositions on the continuous state space $dom(S)$. We can then create
a proposition preserving partition function $\T{0}$, a set of discrete states
$\mathcal{S}^{(0)} = \{ \varsigma_0, \varsigma_1, \dots,$ $ \varsigma_n \}$,
and a set of initial discrete states $\mathcal{S}_\text{init}^{(0)} \subseteq
\mathcal{S}^{(0)}$. The state space $\mathcal V^{(0)}$ and the initial state
$\mathcal V^{(0)}_{init}$ are defined as $\mathcal V^{(0)} = \mathcal
S^{(0)}\times dom(\mathcal E)$ and $\mathcal V^{(0)}_{init} = \mathcal
S^{(0)}_{init}\times dom(\mathcal E)$.


Next, we perform a reachability analysis to establish the transition relations in $\mathbb{D}_{o}^{(0)}$ and $\mathbb{D}_{p}^{(0)}$. For every pair of states, $\nu_a = (\varsigma_a, e_a)$, $\nu_b = (\varsigma_b,e_a)$, we add a transition in $\mathbb{D}^{(0)}_p$ from $\nu_a$ to $\nu_b$ if and only if $\mathcal{R}_p(T^{-1}_{\mathcal{S}^{(0)}}(\varsigma_a), T^{-1}_{\mathcal{S}^{(0)}}(\varsigma_b)) = 1$ and a transition in $\mathbb{D}^{(0)}_o$ if and only if $\mathcal{R}_o(T^{-1}_{\mathcal{S}^{(0)}}(\varsigma_a), T^{-1}_{\mathcal{S}^{(0)}}(\varsigma_b)) = 1$.
\shorten{
\begin{remark}
  $\mathbb{D}_{o}^{(0)}$ is optimistic in the sense that even if only some part of a region corresponding to a discrete state can reach another, we consider there to be a transition between these two discrete states. In $\mathbb{D}_p^{(0)}$ we require every point in a region corresponding to a discrete state to be able to reach to some point in the other for there to be a transition.
\end{remark}
}

\finalcdc{
The idea is illustrated in Figure \ref{fig:construction_of_fts}. Given an
initial proposition preserving partition of the continuous state space (the
colored quardrants), the two FTSs can be constructed using a reachability
analysis. An arrow from a region separated by a solid or dashed line to another
region means that there is some control action taking the system from the first
region to the other. For simplicity, we assume that the environment does not
have any variables.  

\begin{figure}[hbt!]
  \centering
  \begin{tikzpicture}[scale=0.65]

    \draw (0, 0.5) circle (0.3);
    \draw (0, 2.5) circle (0.3);
    \draw (2, 0.5) circle (0.3);
    \draw (2, 2.5) circle (0.3);

    \draw [<-, thick] (0, 0.8) -- (0, 2.2);
    \draw [->, thick] (2, 2.2) -- (2, 0.8);

    \node at (0, 2.5) {$\varsigma_2$};
    \node at (2, 0.5) {$\varsigma_4$};
    \node at (0, 0.5) {$\varsigma_3$};
    \node at (2, 2.5) {$\varsigma_1$};

    \node at (1, 3.3) {$\mathbb{D}_p^{(0)}$};

    \node at (3.5, 1.5) {and};

    \draw (5, 0.5) circle (0.3);
    \draw (5, 2.5) circle (0.3);
    \draw (7, 0.5) circle (0.3);
    \draw (7, 2.5) circle (0.3);

    \draw [<-, thick] (5.3, 0.5) -- (6.7, 0.5);
    \draw [<-, thick] (5, 0.8) -- (5, 2.2);
    \draw [->, thick] (5.3, 2.5) -- (6.7, 2.5);
    \draw [->, thick] (7, 2.2) -- (7, 0.8);

    \node at (5, 2.5) {$\varsigma_2$};
    \node at (7, 0.5) {$\varsigma_4$};
    \node at (5, 0.5) {$\varsigma_3$};
    \node at (7, 2.5) {$\varsigma_1$};

    \node at (6, 3.3) {$\mathbb{D}_o^{(0)}$};


    \fill [pattern=crosshatch, pattern color=green] (-4, 1.5) rectangle (-2.5, 3);
    \fill [pattern=checkerboard, pattern color=black!20] (-4, 1.5) rectangle (-2.5, 0);
    \fill [pattern=north west lines, pattern color=magenta] (-4, 1.5) rectangle (-5.5, 0);
    \fill [pattern=crosshatch dots, pattern color=cyan] (-4, 1.5) rectangle (-5.5, 3);


    \draw [thick, ->] (-3.2, 2.25) -- (-3.2, 1);

    \draw [thick, dashed] (-4, 1.5) -- (-5.5, 3);
    \draw [thick, ->] (-5.25, 2.25) -- (-5.25, 1);
    \draw [thick, ->] (-4.25, 2.25) -- (-4.25, 1);
    \draw [thick, ->] (-4.5, 2.5) -- (-3.75, 2.5); 

    \draw [thick, dashed] (-3.5, 0.5) circle (0.25);
    \draw [thick, ->] (-3.5, 0.5) -- (-4.5, 0.5);

    \draw [-, thick] (-4, 0) -- (-4, 3);
    \draw [-, thick] (-5.5, 1.5) -- (-2.5, 1.5);

    \node at (-1.2, 1.5) {$\Longrightarrow$};

  \end{tikzpicture}
  \caption{Construction of $\mathbb{D}_p^{(0)}$ and $\mathbb{D}_o^{(0)}$ given
an initial proposition preserving partition of the state space (the four
colored quadrants) and a reachability analysis (illustrated by the lines and arrows in the state space). For simplicity, the environment is assumed to have no variables.}
  \label{fig:construction_of_fts}
\end{figure}
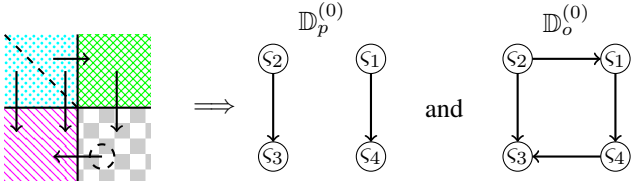
}

We now provide two theorems regarding the (largest) winning sets of
$\mathbb{D}^{(0)}_p$, $\mathbb{D}^{(0)}_o$ and $TS(\Sigma)$, the proofs of
which are reported after the statements of the theorems for the sake of legibility. 
\begin{theorem}
 For any discrete state $\varsigma[0] \in W(\mathbb D^{(0)}_p,\varphi)$ that is winning for the pessimistic FTS $\Dp{0}$, the corresponding continuous state is also winning in $TS(\Sigma)$, i.e., $\invT{0}(\varsigma[0])$ $\subseteq W(TS(\Sigma),\varphi)$.
  \label{thrm:pes2cont0}
\end{theorem}
\begin{theorem}
 For any continuous state $s[0] \in W(TS(\Sigma),\varphi)$ that is winning for $TS(\Sigma)$, the corresponding discrete state is also winning in $\mathbb D^{(0)}_o$, i.e., $\T{0}(s[0])\in W(\mathbb D^{(0)}_o,\varphi)$.
  \label{thrm:cont2opt0}
\end{theorem}
\begin{proof}[Proof of Theorem~\ref{thrm:pes2cont0}]
  Suppose the winning controller for $W(\mathbb D^{(0)}_p,\varphi)$ is $\gamma_p = (\gamma_{p, 1}, \gamma_{p, 2}, \dots, \gamma_{p, t}, \dots)$. Consider a discrete state $\varsigma[0] = \T{0}(s[0]) \in  W(\mathbb D^{(0)}_p,\varphi)$. For all possible environmental actions $e[0]e[1]\dots$, we can create the controlled execution using $\gamma_p$. This gives a sequence of states $(\varsigma[0],e[0])(\varsigma[1],e[1])\dots$, which satisfies the specification $\varphi$. 

  Consider now a continuous state $s[0]\in \invT{0}(\varsigma[0])$. From the construction of $\mathbb D^{(0)}_p$, we know that
  \begin{displaymath}
    \mathcal{R}_p(T^{-1}_{\mathcal{S}^{(0)}}(\varsigma[t]), T^{-1}_{\mathcal{S}^{(0)}}(\varsigma[t+1])) = 1.
  \end{displaymath}
  Thus, we can recursively define the consistent continuous controller $\gamma
= (\gamma_1, \gamma_2, \dots)$ to be
  \begin{enumerate}
    \item $\gamma_1(s[0],e[0])$ returns an $s[1]\in \invT{0}(\varsigma[1])$ such that there exists an $u[0]\in U$ and $f(s[0],u[0]) = s[1]$. 
    \item $\gamma_{t+1}(s[0],e[0],\dots,e[t])$ returns an $s[t+1]\in \invT{0}(\varsigma[t+1])$ such that there exists an $u[t]\in U$ and 
      \begin{align*}
          f( \gamma_{t}(s[0],e[0],\dots&,e[t-1]),u[t])\\
	&=\gamma_{t+1}(s[0],e[0],\dots,e[t]).
      \end{align*}
  \end{enumerate}
  As a result, we have a sequence $(s[0],e[0])(s[1],e[1])\dots$, where $\T{0}(s[t]) = \varsigma[t]$. Hence, the controller $\gamma$ is also winning at $s[0]$, which completes the proof.
\end{proof}

\begin{proof}[Proof of Theorem~\ref{thrm:cont2opt0}]

  Suppose $\gamma = (\gamma_1, \gamma_2, \dots)$ is winning for $W(TS(\Sigma),\varphi)$ and $s[0]\in W(TS(\Sigma),\varphi)$. For all possible environmental actions $e[0]e[1]\dots$, we create a controlled execution using $\gamma$: $(s[0],e[0])(s[1],e[1])\dots$, which is winning.
 
  Now consider the discrete state $\varsigma[t] = \T{0}(s[t])$. By the definition of $\mathcal R_o$, we know that
  \begin{displaymath}
    (\varsigma[t],e[t]) \rightarrow_{o}^{(0)}(\varsigma[t+1],e[t+1]).
  \end{displaymath}
  As a result, we can construct a consistent controller  $\gamma_o
  = (\gamma_{o,1},\dots)$ for $\varsigma[0] = \T{0}(s[0])$ as
  $\gamma_{o,t}(\varsigma[0],e[0], \dots, e[t-1])
  = \T{0}\left(\gamma_t(s[0],e[0],\dots,e[t-1]\right)$.
  Thus, we get a sequence $(\varsigma[0],e[0])(\varsigma[1],e[1])\dots$, where $\varsigma[t] = \T{0}(s[t])$. Hence, the controller $\gamma_o$ is winning at $\varsigma[0]$, which completes the proof.
\end{proof}

We now define the following three sets:
\begin{equation}
  \mathcal{W}^{(i)} = W(\mathbb{D}_p^{(i)}, \varphi),
\end{equation}
referred to as \emph{the winning set};
\begin{equation}
  \mathcal{L}^{(i)} = L(\mathbb{D}_o^{(i)}, \varphi)
\end{equation}
as \emph{the losing set}; and
\begin{equation}
  \mathcal{M}^{(i)} =\mathcal S^{(i)}\backslash \left(\mathcal{W}^{(i)} \cup
  \mathcal{L}^{(i)}\right),
\end{equation}
as the \emph{the maybe set}. We can further define the inverse image of these
sets on $dom(S)$ as $\mathcal W_c^{(i)}
= \invT{i}(\mathcal W^{(i)})$, $\mathcal L_c^{(i)} = \invT{i}(\mathcal
L^{(i)})$ and $\mathcal M_c^{(i)} = \invT{i}(\mathcal M^{(i)})$.

By Theorem~\ref{thrm:pes2cont0} and \ref{thrm:cont2opt0}, it is clear that
\vskip 1mm
\fbox{\begin{minipage}{0.93\columnwidth}
\begin{enumerate}
  \item If $S_{init}\subseteq\mathcal W_c^{(0)}$, then $S_{init}$ is a winning set for $TS(\Sigma)$. Furthermore, the winning controller can be constructed in a similar fashion as is discussed in the proof of Theorem~\ref{thrm:pes2cont0}.
  \item If $S_{init}\bigcap \mathcal L_c^{(0)}\neq \emptyset$, then $S_{init}$ is not a winning set for $TS(\Sigma)$. 
  \item If neither 1) nor 2) is true, then a finer partition is needed to answer the Realizability Problem.
\end{enumerate}
\end{minipage}}
\vskip 2mm

For case 3), one may naively create a finer partition function and the
corresponding pessimistic and optimistic FTSs. In the next subsection, we show
how to iteratively do this in order to reduce the computational complexity of
the abstraction algorithm by exploiting the properties of the winning set.

\subsection{Refinement Procedure}

We define a refinement operation as
\begin{equation}
  \mathbf{split}_m:2^{dom(S)} \times \{1, \dots, m \} \rightarrow 2^{dom(S)}
\end{equation}
such that for all $X \subseteq dom(S)$ and $i, j \in \{1, \dots, m\}, i \neq
j$, it has the following properties: $\mathbf{split}_m(X, i) \subset X$,
$\mathbf{split}_m(X, i) \cap \mathbf{split}_m(X, j) = \emptyset$ and
$\bigcup\limits_{k = 1}^m \mathbf{split}_m(X, k) = X$.

\begin{remark}
  The index $m$ on $\mathbf{split}_m$ is the number of children that a region
  should be split into upon refinement. We leave it unspecified how to choose
  $m$ and the exact shape of the regions generated by $\mathbf{split}_m$, since
  the exact details are not relevant for the algorithm. In the implementation
  in Section~\ref{sec:simulation}, a split of $X \subset \mathbb{R}^n$ into
  $2^n$ equally sized hyperrectangles was used (assuming that the initial
  proposition preserving partition consisted only of hyperrectangles).
\end{remark}

We will focus our computational resources (i.e. perform a further refinement) on the states in the maybe set
$\mathcal{M}^{(i)}$. Intuitively, these states have the potential to become
winning when we create finer partitions. With $\mathcal S^{(i)}$ and $\T{i}$ as
the set of discrete states and the partition function of the $i$th iteration,
respectively, we define $\mathcal S^{(i+1)}$ and $\T{i+1}$ in the following way:

\begin{enumerate}
  \item If $\varsigma\in \mathcal W^{(i)}\cup \mathcal L^{(i)}$, then $(\varsigma,1)\in \mathcal S^{(i+1)}$ and
    \begin{displaymath}
      \invT{i+1}( (\varsigma,1)) = \invT{i}(\varsigma).
    \end{displaymath}
 \item If $\varsigma\in \mathcal M^{(i)}$, then $(\varsigma,j)\in \mathcal S^{(i+1)}$ for all $j = 1,\dots,m$ and
    \begin{displaymath}
      \invT{i+1}( (\varsigma,j)) = \mathbf{split}_m(\invT{i}(\varsigma),j).
    \end{displaymath}
\end{enumerate}
Given the discrete states, the state space $\mathcal V^{(i+1)}$ can be defined as
  $\mathcal V^{(i+1)} = \mathcal S^{(i+1)}\times dom(\mathcal E)$, 
and the initial states $\mathcal V^{(i+1)}_{init}$ can be defined in a similar fashion.

\shorten{
\begin{remark}
One can consider the discrete state spaces $\mathcal S^{(0)}$, $\mathcal
S^{(1)}$, $\dots$ to form a forest (a disjoint union of trees), where the states in $\mathcal S^{(0)}$ are the roots and $(\varsigma,j)\in \mathcal S^{(i+1)}$ is the $j$th child of $\varsigma\in \mathcal S^{(i)}$. 
\end{remark}
}

\finalcdc{
A simple example of the refinement procedure is provided in
Figure~\ref{fig:refinement}. An initial preposition preserving partition is
constructed from the continuous state space $dom(S)$, which in this case,
results in three discrete states (and corresponding regions in the continuous
state space). The discrete states are marked as to belonging to either the
winning (crosshatched green), maybe (solid yellow) or losing (dotted red) set.
To refine the partition, the $\mathbf{split}_3$-operator (using equally sized
rectangles as partitions) is applied to the state in the maybe set,
namely $\varsigma_2$. The refined partition can be seen in the rightmost figure, where
a new reachability analysis has been performed. The next step of the procedure
would further refine the new maybe set, $(\varsigma_2,3)$.

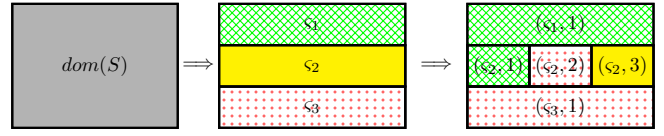
\begin{figure}[hbt!]
  \centering
  \begin{tikzpicture}[thick, scale=0.55, every node/.style={scale=0.75}]

    
    \filldraw [draw=black, pattern=dots, pattern color=red!60] (0, 0)
rectangle (4.5, 1);
    \filldraw [draw=black, fill=yellow] (0, 1)
rectangle (4.5, 2);
    \filldraw [draw=black, pattern=crosshatch, pattern color=green] (0, 2)
rectangle (4.5, 3);

    \filldraw [draw=black, pattern=dots, pattern color=red!60] (6, 0)
rectangle (10.5, 1);
    \filldraw [draw=black, fill=yellow] (9, 1)
rectangle (10.5, 2);
    \filldraw [draw=black, pattern=crosshatch, pattern color=green] (6, 2)
rectangle (10.5, 3);
    \filldraw [draw=black, pattern=crosshatch, pattern color=green] (6, 1)
rectangle (7.5, 2);
    \filldraw [draw=black, pattern=dots, pattern color=red!60] (7.5, 1)
rectangle (9, 2);

    \node at (0.75 + 6, 1.5) {$(\varsigma_2, 1)$};
    \node at (2.25 + 6, 1.5) {$(\varsigma_2, 2)$};
    \node at (3.75 + 6, 1.5) {$(\varsigma_2, 3)$};

    \node at (6 + 2.25, 0.5) {$(\varsigma_3, 1)$};
    \node at (6 + 2.25, 2.5) {$(\varsigma_1, 1)$};
    
    \node at (2.25, 0.5) {$\varsigma_3$};
    \node at (2.25, 1.5) {$\varsigma_2$};
    \node at (2.25, 2.5) {$\varsigma_1$};

    \draw [draw=black] (0, 0) rectangle (4.5, 3);
    \draw [draw=black] (0, 0) rectangle (4.5, 1);
    \draw [draw=black] (0, 1) rectangle (4.5, 2);

    \draw [draw=black] (6, 0) rectangle (10.5, 3);
    \draw [draw=black] (6, 0) rectangle (10.5, 1);
    \draw [draw=black] (6, 1) rectangle (10.5, 2);
    \draw [draw=black] (6, 1) rectangle (7.5, 2);
    \draw [draw=black] (7.5, 1) rectangle (9, 2);

    \filldraw [draw=black, fill=black!30] (-5.0, 0) rectangle (-1.0, 3);
    \draw [draw=black] (-5, 0) rectangle (-1.0, 3);
    \node at (-3, 1.5) {$dom(S)$};

    \node at (5.25, 1.5) {$\Longrightarrow$};
    \node at (-0.5, 1.5) {$\Longrightarrow$};


  \end{tikzpicture}
  \caption{An example of the proposed refinement procedure. An initial
preposition preserving partition is constructed in the first step. The regions
are labeled with their corresponding discrete state. The states are colored
differently depending on if they belong to the winning (crosshatched green),
maybe (solid yellow) or losing (dotted red) set. The
$\mathbf{split}_3$-operator is used to further refine the states in the maybe
set (only one iteration is illustrated).}
  \label{fig:refinement}
\end{figure}
}

 We now define the transition relations of the two FTSs. We begin with the relations in the pessimistic FTS. 
For any two states $(\varsigma_a, j) , (\varsigma_b, k) \in \mathcal{S}^{(i+1)}$ and environmental states $e_a$, $e_b$, we have that $((\varsigma_a, j),e_a) \rightarrow_p^{(i+1)} ((\varsigma_b, k),e_b)$ if and only if one of the following statements holds:

\begin{enumerate}
  \item \emph{WW-transition}: 
    $\varsigma_a, \varsigma_b \in \mathcal{W}^{(i)}$, $j = k = 1$ and 
    \begin{displaymath}
      (\varsigma_a,e_a)\rightarrow_p^{(i)}(\varsigma_b,e_b).
    \end{displaymath}
  \item
\emph{MW-transition}:
    $\varsigma_a \in \mathcal{M}^{(i)}$, $\varsigma_b \in \mathcal{W}^{(i)}$, $k =1$ and 
    \begin{displaymath}
      \mathcal{R}_p( \, T^{-1}_{\mathcal{S}^{(i+1)}}( \, ( \varsigma_a, j) \, ), \, T^{-1}_{\mathcal{S}^{(i+1)}}( \, (\varsigma_b, 1) \, ) \, ) = 1.
    \end{displaymath}
  \item
\emph{MM-transition}:
    $\varsigma_a, \varsigma_b \in \mathcal{M}^{(i)}$ 
    and 
    \begin{displaymath}
     \mathcal{R}_p( \, T^{-1}_{\mathcal{S}^{(i+1)}}( \, (\varsigma_a, j) \, ), \, T^{-1}_{\mathcal{S}^{(i+1)}}( \, (\varsigma_b, k) \, ) \, ) = 1. 
    \end{displaymath}
\end{enumerate}
\begin{remark}
WW stands for a transition between two winning states, and analogously for MW
and MM.  Notice that we omit many possible transitions. This allows us to focus
on the critical transitions that affects the computation of the winning set.
\shorten{The rationale for this is that it is waste to check if, for example,
a winning state can reach a maybe state, since we already know that there is
a winning controller in the winning state.}
\end{remark}

The update rule for the optimistic FTS is similar. We have that $((\varsigma_a, j),e_a) \rightarrow_o^{(i+1)} ((\varsigma_b, k),e_b)$ if and only if one of the following three statements holds:
\begin{enumerate}
  \item \emph{WW-transition}: $\varsigma_a, \varsigma_b \in \mathcal{W}^{(i)}$, $j = k = 1$ and 
    \begin{displaymath}
      (\varsigma_a,e_a)\rightarrow_p^{(i)}(\varsigma_b,e_b).
    \end{displaymath}
    Notice that we are using the transition relation $\rightarrow_p^{(i)}$ instead of $\rightarrow_o^{(i)}$ for this case.

  \item \emph{MW-transition}: $\varsigma_a \in \mathcal{M}^{(i)}$, $\varsigma_b \in \mathcal{W}^{(i)}$, $k =1$ and 
    \begin{displaymath}
      \mathcal{R}_o( \, T^{-1}_{\mathcal{S}^{(i+1)}}( \, ( \varsigma_a, j) \, ), \, T^{-1}_{\mathcal{S}^{(i+1)}}( \, (\varsigma_b, 1) \, ) \, ) = 1.
    \end{displaymath}
  \item \emph{MM-transition}: $\varsigma_a, \varsigma_b \in \mathcal{M}^{(i)}$ and 
    \begin{displaymath}
     \mathcal{R}_o( \, T^{-1}_{\mathcal{S}^{(i+1)}}( \, (\varsigma_a, j) \, ), \, T^{-1}_{\mathcal{S}^{(i+1)}}( \, (\varsigma_b, k) \, ) \, ) = 1. 
    \end{displaymath}
\end{enumerate}



We will now expand upon Theorem~\ref{thrm:pes2cont0} and \ref{thrm:cont2opt0}
to provide a characterization of the winning sets $W(\mathbb
D_p^{(i)},\varphi)$ and $W(\mathbb D_o^{(i)},\varphi)$. \finalcdc{The proofs of
    the following theorems are deferred to the appendix for the sake of
legibility.}

\begin{theorem}
  For any discrete state $\varsigma[0] \in W(\mathbb D^{(i)}_p,\varphi)$ that is winning for the pessimistic FTS $\Dp{i}$, the corresponding continuous state is also winning in $TS(\Sigma)$, i.e., 
 \begin{equation}
  \invT{i}(\varsigma[0])\subseteq W(TS(\Sigma),\varphi).  
   \label{eq:setinclusionpes}
 \end{equation}
 Furthermore, its child $(\varsigma[0],1)$ is also winning for $\mathbb D^{(i+1)}_p$, i.e., 
 \begin{equation}
   (\varsigma[0],1)\in W(\Dp{i+1},\varphi).  
   \label{eq:setinclusionpes2}
 \end{equation}
  \label{thrm:pes2cont}
\end{theorem}
\begin{theorem}
 For any continuous state $s[0] \in W(TS(\Sigma),\varphi)$ that is winning for $TS(\Sigma)$, the corresponding discrete state is also winning in $\mathbb D^{(i)}_o$, i.e., 
 \begin{equation}
 \T{i}(s[0])\in W(\mathbb D^{(i)}_o,\varphi).
   \label{eq:setinclusionopt}
 \end{equation}
 Furthermore, if the discrete state $\varsigma[0]\in L(\Do{i},\varphi)$ is losing for $\Do{i}$, then its child is also losing in $\Do{i+1}$, i.e.,
 \begin{equation}
   (\varsigma[0],1)\in L(\Do{i+1},\varphi).
   \label{eq:setinclusionopt2}
 \end{equation}

  \label{thrm:cont2opt}
\end{theorem}
Combining Theorem~\ref{thrm:pes2cont} and \ref{thrm:cont2opt}, we have the following corollary:
\begin{corollary}
  $\mathcal W_c^{(0)}\subseteq \mathcal W_c^{(1)}\subseteq\dots\subseteq
W(TS(\Sigma),\varphi)\subseteq \dots \subseteq dom(S)\;\backslash\;\mathcal
L_c^{(1)}\subseteq dom(S)\;\backslash\;\mathcal L_c^{(0)}$.
  \label{crly:cont2opt}
\end{corollary}

The box outlining the algorithm for the first iteration can be
straight-forwardly adjusted with Theorem~\ref{thrm:pes2cont} and
\ref{thrm:cont2opt} to outline the full algorithm.
\shorten{
\vskip 1mm
\fbox{\begin{minipage}{0.93\columnwidth}
\begin{enumerate}
  \item If $S_{init}\subseteq\mathcal W_c^{(i)}$, then $S_{init}$ is a winning
set for $TS(\Sigma)$. A winning controller can be constructed in a similar fashion as is discussed in the proof of Theorem~\ref{thrm:pes2cont0}.
  \item If $S_{init}\bigcap \mathcal L_c^{(i)}\neq \emptyset$, then $S_{init}$
is not a winning set for $TS(\Sigma)$. Thus, we can stop the refinement
procedure because there is no winning controller.
  \item If neither of the above statements is fulfilled, then we cannot give a definitive answer on whether $S_{init}$ is winning or not at the $i$th iteration. As a result, we create the FTSs $\Dp{i+1}$ and $\Do{i+1}$ and try to solve the winning sets for them.
\end{enumerate}
\end{minipage}}
\vskip 2mm
}

\shorten{
\begin{remark}
 It is worth noticing that we do not use any special properties of the $f$ function or
the sets $U,\, dom(S)$ and $S_{init}$, except for the reachability relations
that they induce. As a result, the algorithm presented in this article can be used to handle any transition system.
\end{remark}
}

\section{Numerical Results}
\label{sec:simulation}

\commentsr{Review: Use past tense}

\commentsr{Review: Could you explain why the threshold 0.2 is relevant here?}

\comments{Robert: Mention that only two dimensional systems are considered because they can be visualized and for simplicity and that our implementation handles splitting using hyperrectangles in higher dimensions}

\comments{Yilin: Fair enough.}

In this section, we perform a comparison between the algorithm in
\texttt{TuLiP} \cite{tulip} and our proposed algorithm on two systems in
$\mathbb{R}^2$ (for simplicity and illustratory purposes, the algorithm is
valid for higher-dimensional systems as well). All the simulations were
performed on a MacBook Air (1.3 GHz, 4 GB RAM).

\comments{
\subsection{Temporal Improvements}
}

Consider the system
\begin{equation}
  \begin{split}
    s[t+1] &= I_2 s[t] +
              I_2 u[t], \\
    u[t] &\in U = \{ v \in \mathbb{R}^2 : |v|_{\infty} \leq 1 \}, \\
    s[t] &\in dom(S) = [0,3] \times [0,2], \\
    s[0] &\in S_\text{init} = [0, 3] \times [0, 2],
  \end{split}
\end{equation}
where $I_2$ is the identity matrix with two columns, with the following
propositional markings in the state space: $[0,1] \times [0,1]$ as \emph{home}
and $[2,3] \times [1,2]$ as \emph{lot}. Let the environment be equipped with
a Boolean variable, \emph{park}, and let the specification of system be the
following:
  $\varphi = \square\lozenge
  \emph{home}\wedge\square(\emph{park}\rightarrow\lozenge\emph{lot})$,
which can be converted into GR1-form. \shorten{Roughly speaking the
specification implies that the system should visit the parking \emph{lot}
whenever the environment sets \emph{park} true, and always returns back
\emph{home}.}

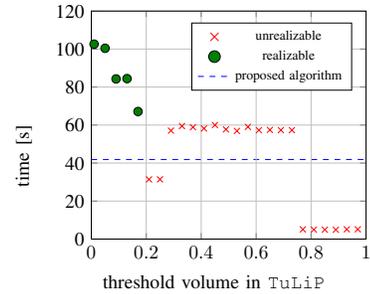
\begin{figure}[b!]
  \centering
\begingroup
\tikzset{every picture/.style={scale=0.73}}
%
%
%
%
\begin{tikzpicture}

\begin{axis}[
xlabel={$\quad\quad$ threshold volume in \texttt{TuLiP}},
ylabel={time [s]},
xmin=0, xmax=1,
ymin=0, ymax=120,
axis on top,
ytick={0,20,40,60,80,100,120},
yticklabels={0,20,40,60,80,100,120},
xmajorgrids,
ymajorgrids,
legend style={at={(1.3,1.3)}, anchor=north east, font=\scriptsize},
legend entries={{unrealizable},{realizable},{proposed algorithm}}
]
\addplot [red, mark=x, mark size=3, only marks]
coordinates {
(0.21,31.39)
(0.25,31.43)
(0.29,57.07)
(0.33,59.47)
(0.37,58.86)
(0.41,58.19)
(0.45,60.02)
(0.49,57.71)
(0.53,56.91)
(0.57,59.05)
(0.61,57.29)
(0.65,57.45)
(0.69,57.33)
(0.73,57.31)
(0.77,5.1)
(0.81,4.96)
(0.85,4.94)
(0.89,4.9)
(0.93,5.14)
(0.97,5.13)

};
\addplot [green!50.0!black, mark=*, mark size=3, mark options={draw=black}, only marks]
coordinates {
(0.01,102.57)
(0.05,100.4)
(0.09,84.27)
(0.13,84.43)
(0.17,67.12)

};
\addplot [blue, dashed]
coordinates {
(0,41.88)
(0.0204081632653061,41.88)
(0.0408163265306122,41.88)
(0.0612244897959184,41.88)
(0.0816326530612245,41.88)
(0.102040816326531,41.88)
(0.122448979591837,41.88)
(0.142857142857143,41.88)
(0.163265306122449,41.88)
(0.183673469387755,41.88)
(0.204081632653061,41.88)
(0.224489795918367,41.88)
(0.244897959183673,41.88)
(0.26530612244898,41.88)
(0.285714285714286,41.88)
(0.306122448979592,41.88)
(0.326530612244898,41.88)
(0.346938775510204,41.88)
(0.36734693877551,41.88)
(0.387755102040816,41.88)
(0.408163265306122,41.88)
(0.428571428571429,41.88)
(0.448979591836735,41.88)
(0.469387755102041,41.88)
(0.489795918367347,41.88)
(0.510204081632653,41.88)
(0.530612244897959,41.88)
(0.551020408163265,41.88)
(0.571428571428571,41.88)
(0.591836734693878,41.88)
(0.612244897959184,41.88)
(0.63265306122449,41.88)
(0.653061224489796,41.88)
(0.673469387755102,41.88)
(0.693877551020408,41.88)
(0.714285714285714,41.88)
(0.73469387755102,41.88)
(0.755102040816326,41.88)
(0.775510204081633,41.88)
(0.795918367346939,41.88)
(0.816326530612245,41.88)
(0.836734693877551,41.88)
(0.857142857142857,41.88)
(0.877551020408163,41.88)
(0.897959183673469,41.88)
(0.918367346938775,41.88)
(0.938775510204082,41.88)
(0.959183673469388,41.88)
(0.979591836734694,41.88)
(1,41.88)

};
\path [draw=black, fill opacity=0] (axis cs:0,120)--(axis cs:1,120);

\path [draw=black, fill opacity=0] (axis cs:1,1.77635683940025e-15)--(axis cs:1,120);

\path [draw=black, fill opacity=0] (axis cs:0,1.77635683940025e-15)--(axis cs:1,1.77635683940025e-15);

\path [draw=black, fill opacity=0] (axis cs:0,1.77635683940025e-15)--(axis cs:0,120);

\end{axis}

\end{tikzpicture}
\endgroup
  \caption{Timing data for the current algorithm in \texttt{TuLiP} and our
proposed algorithm. The specifications that we are considering for the
continuous system are realizable, but \texttt{TuLiP} cannot synthesize
a controller until the threshold volume is below 0.2. The dots and crosses
indicate the time for \texttt{TuLiP} to partition the state space and then try
to synthesize a controller, giving a positive or a negative answer,
respectively, on whether the specifications are realizable. Our algorithm
concludes that the specifications are realizable without taking any threshold
volume as input, illustrated by the dashed blue line.}
  \label{fig:continuous_plot}
\end{figure}

The algorithm employed by \texttt{TuLiP} \cite{rhtlp} partitions the whole
state space according to a reachability analysis until no region corresponding
to a discrete state can be refined further without going below a pre-specified
threshold volume. This leads to problems when the threshold volume is set too
high, since not enough transitions can be established in the finite state
model. As illustrated by the red crosses in Figure \ref{fig:continuous_plot},
\texttt{TuLiP} failed to find a controller realizing the specification when the
threshold volume was taken larger than 0.2. When the threshold was chosen below
this value, it succeeded in finding a controller and announced that the
specifications were realizable (green dots).

Our implementation iteratively refines the partition of the state space until
a controller can be synthesized (or, in the case that the specifications are
unrealizable, until it can guarantee that none can be found). Furthermore, our
algorithm only refines the ``interesting'' areas of the state space, which
results in less computational time -- indicated by the dashed blue line.
\finalcdc{Note that the time it took to ``guess'' the right threshold value for
\texttt{TuLiP} is large.}

\comments{
\subsection{Non-Existence Proof}
}

The next example shows the actual partition that results from the two methods.
Consider the system
\begin{equation}
  \begin{split}
    s[t+1] &= I_2 s[t] +
              I_2 u[t], \\
    u[t] &\in U = \{ v \in \mathbb{R}^2 : |v|_{\infty} \leq 1 \}, \\
    s[t] &\in dom(S) = [0,4] \times [0,4], \\
    s[0] &\in S_\text{init} = [3, 3.5] \times [3, 3.5],
  \end{split}
  \label{eq:example_2}
\end{equation}
with the set of propositions: $[0, 0.5] \times [0, 0.5]$ as \emph{goal} and
$[3, 3.5] \times [3, 3.5]$ as \emph{start}. For simplicity, assume that the
environment has no variables. The initial assumption on the system is
\emph{start} and the progress specification of the system is $\square \lozenge
\text{\emph{goal}}$. \shorten{This means that the systems starts in
\emph{start} and should always eventually reach \emph{goal}. }

A set $\Omega$ is \emph{invariant} if $s(t_0) \in \Omega \Rightarrow \; s(t)
\in \Omega, \; \forall t \geq t_0$ and for all possible controls $u(t)$. It is
simple to show that the region $\mathbb{R}^2 \; \backslash \; [0, 2]^2$ is
invariant for \eqref{eq:example_2}. Since \emph{start} lies in an invariant
region, that does not contain \emph{goal}, we know a priori that there
does not exist a winning controller.

 \begin{figure}
    \subfigure[]{
            \includegraphics[width=0.22\textwidth]{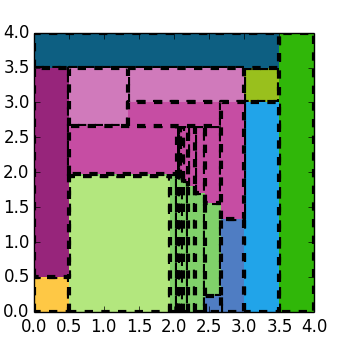}
    }
    \subfigure[]{
            \includegraphics[width=0.215\textwidth]{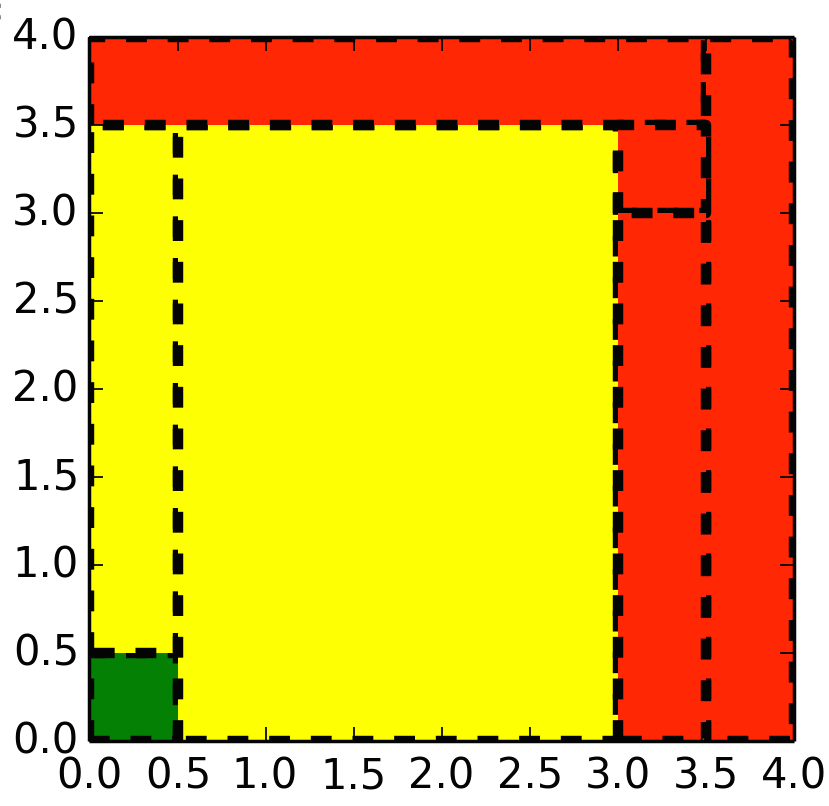}
    }
    \caption{(a) shows the partition by \texttt{TuLiP} of the system
        \eqref{eq:example_2} when the threshold volume was chosen to be 1.0.
        Regions of the same color are considered as one discrete state. (b)
        shows the partition resulting from our algorithm, with the winning
    (green), maybe (yellow) and losing (red) sets marked. Here, every region is
its own discrete state.}
    \label{fig:discretization_example}
 \end{figure}

Figure~\ref{fig:discretization_example}a) shows the partition that
\texttt{TuLiP} provided when the threshold volume was set to 1.0. Note that the
invariant region is finely partitioned. The runtime of the algorithm was 620 s.
No controller that fulfills the specifications could be synthesized using this
abstraction. Note that from the output of \texttt{TuLiP}, it is not possible to
say whether no winning controller exists, or if a winning controller of the
original system exists but \texttt{TuLiP} cannot find it because of the
partition being too coarse.

The output of our algorithm can be seen in
Figure~\ref{fig:discretization_example}b). The coloring illustrates the winning
(green), maybe (yellow) and losing (red) states. The states in the maybe set
are marked as such since some of the continuous states in them lie within the
invariant region, and some lie within the region that can reach \emph{goal}.
Since \emph{start} lies in the losing set, the algorithm terminates and
concludes with a definitive answer that there exists no winning controller
(neither for the abstraction nor the original system). This took 25 s.

\section{Conclusion}
\label{sec:conclusion}

In this paper we have presented an iterative method for abstracting
a discrete-time control system into two FTSs, representing an under- and
over-approximation of the reachability properties of the original dynamical
system.  We have provided theorems regarding the existence of controllers
fulfilling GR1-specifications for the continuous system, based on
the existence of such controllers for the two FTSs. Our proposed algorithm
provides a way of focusing the computational resources on refining only certain
areas of the state space, leading to a decrease in the time complexity of the
abstraction procedure compared to previous methods. We have made a comparison
between the proposed algorithm and the one currently used in the
\texttt{TuLiP}-framework on numerical examples with promising results.

\commentsr{
The method presented in the paper is using under- and over-approximation,
where, if a winning controller is present in the under-approximation, it is
guaranteed to also exist in the real system (and oppositely for the
over-approximation). This creates a gap in between the two approximations where
states are neither surely winning nor surely losing - hence the maybe set. It
has been shown how the abstraction of the discrete-time dynamical system can be
refined to a new over and under approximation while reusing the old computation
results, leaving the algorithm to focus only on refining and computing the
maybe-set.
}

\bibliographystyle{abbrv}
\bibliography{surf_ref}

\finalcdc{
\appendix

\setlength{\belowdisplayskip}{4pt} \setlength{\belowdisplayshortskip}{4pt}
\setlength{\abovedisplayskip}{3pt} \setlength{\abovedisplayshortskip}{3pt}

\shorten{
\begin{proof}[Proof of Proposition~\ref{prop:union_of_W}]
  Let us define an index function $h:\bigcup_{i\in\mathcal I} \mathcal W_i\rightarrow \mathcal I$, such that for any $\varsigma\in \bigcup_{i\in\mathcal I} \mathcal W_i$, the following set inclusion holds:
  \begin{displaymath}
    \varsigma \in \mathcal W_{h(\varsigma)}. 
  \end{displaymath}
   Now assume that the winning controller for the set $\mathcal W_i$ is
$\gamma^{(i)} = ( \gamma_1^{(i)}, \gamma_2^{(i)}, \dots, \gamma_t^{(i)}, \dots
)$. We can define the new controller $\gamma = (\gamma_1,\gamma_2,\dots)$ as
   \begin{displaymath}
     \gamma_t (\varsigma[0],e[0],\dots,e[t-1]) = \gamma_t^{(h(\varsigma[0]))}(\varsigma[0],e[0],\dots,e[t-1]).
   \end{displaymath}
   It is easily verified that $\gamma$ is a winning controller for \\ $\bigcup_{i\in\mathcal I}\mathcal W_i$.
\end{proof}
}

\begin{lemma}
  For any two sequences $\sigma = \nu_0\nu_1\dots$, $\sigma' = \nu_0'\nu_1'\dots$, such that $\sigma\models\varphi$ and $\sigma'\models \varphi$, where $\varphi$ is a GR1 formula defined in \eqref{eq:gr1}, the following properties hold:
  \begin{enumerate}
    \item Define a time-shifted sequence $\sigma_t = \nu_t\nu_{t+1}\dots$, then $\sigma_t \models \varphi$. 
    \item Suppose that there exists $\tau\geq 0$, such that $\nu_\tau = \nu'_0$, then the following sequence $\nu_0\dots\nu_\tau\nu'_1\nu'_2\dots\models \varphi$.
  \end{enumerate}
  \label{lemma:timeinvariant}
\end{lemma}
\begin{proof}
By definition, $\sigma \models \varphi$ if and only if
\begin{equation}
 \sigma \models \left(\neg \bigwedge_{i=1}^M \square\lozenge p_i\right) \bigvee\left( \bigwedge_{j=1}^N \square\lozenge q_j\right).
 \label{eq:gr1transform}
\end{equation}
The lemma follows directly from the fact that the right hand side of \eqref{eq:gr1transform} is a liveness formula. 
\end{proof}
\begin{lemma}
  Consider an FTS $\mathbb T$ and a GR1 formula $\varphi$. If the controller $\gamma$ is winning for some non-empty set $\mathcal W$, then for any initial condition $\varsigma[0]\in \mathcal W$ and environmental actions $e[0]e[1]\dots$, the controlled execution $(\varsigma[0],e[0])(\varsigma[1],e[1])\dots$ satisfies
  \begin{displaymath}
   \varsigma[t]\in W(\mathbb T,\varphi) , \,\forall t=0,1,\dots \; .
  \end{displaymath}
  \label{lemma:setinvariant}
\end{lemma}
\begin{proof}
This result follows directly from Lemma~\ref{lemma:timeinvariant}.
\end{proof}
\begin{proof}[Proof of Theorem~\ref{thrm:pes2cont}]
By the recursive definition of $\mathbb D^{(i)}_p$ and $\mathbb D^{(i)}_o$, we know that for any $\varsigma_a,\varsigma_b\in \mathcal S^{(i)}$,
\begin{displaymath}
  (\varsigma_a,e_a)\rightarrow_p^{(i)}(\varsigma_b,e_b)
\end{displaymath}
implies that 
\begin{displaymath}
  \mathcal{R}_p(  \invT{i} ( \varsigma_a) ), \invT{i}(\varsigma_b))) = 1.
\end{displaymath}
Hence, \eqref{eq:setinclusionpes} can be proved in a similar way as Theorem~\ref{thrm:pes2cont0}.

We now prove \eqref{eq:setinclusionpes2}. For the FTS $\mathbb D_p^{(i)}$,
suppose the winning controller for $\mathcal W^{(i)} = W(\mathbb
D^{(i)}_p,\varphi)$ is $\gamma_p^{(i)}
= (\gamma_{p,1}^{(i)},\gamma_{p,2}^{(i)},\dots)$. We can define the controller
$\gamma_p^{(i+1)} = (\gamma_{p,1}^{(i+1)},\gamma_{p,2}^{(i+1)},$ $\dots)$ for
the FTS $\mathbb D_p^{(i+1)}$ as
\begin{align*}
  \gamma_{p,t}^{(i+1)}((\varsigma[0],1),&e[0],\dots,e[t-1]) \\
  &= (\gamma_{p,t}^{(i)}(\varsigma[0],e[0],\dots,e[t-1]) ,1).
\end{align*}
Thus, the controlled execution of the FTS $\Dp{i+1}$ is given by
\begin{displaymath}
  ((\varsigma[0],1),e[0]) ((\varsigma[1],1),e[1]) ((\varsigma[2],1),e[2])\dots,
\end{displaymath}
which satisfies the specification $\varphi$. Therefore, we only need to prove
that the controller $\gamma_p^{(i+1)}$ is consistent.

By Lemma~\ref{lemma:setinvariant}, we know that for any $\varsigma[0]\in \mathcal W^{(i)}$, the controlled execution $(\varsigma[0],e[0])\dots$ satisfies
\begin{displaymath}
  \varsigma[t]\in \mathcal W^{(i)},  
\end{displaymath}
which implies that the transition from $((\varsigma[t],1),e[t])$ to
$((\varsigma[t+1],1),e[t+1])$ in $\Dp{i+1}$ is a \emph{WW-transition} and hence
exists. Hence, $\gamma_p^{(i+1)}$ is consistent, which completes the proof.
\end{proof}

\begin{proof}[Proof of Theorem~\ref{thrm:cont2opt}]
  We first prove \eqref{eq:setinclusionopt2}. Notice that by the construction of $\Do{i+1}$, if $\varsigma[0] \in \mathcal L^{(i)} = L(\Do{i},\varphi)$, then $((\varsigma[0],1),e[0])$ has no successors in $\Do{i+1}$. Thus, $(\varsigma[0],1)\in L(\Do{i+1},\varphi)$ since no consistent controller exists for $(\varsigma[0],1)$.

  We now prove \eqref{eq:setinclusionopt} by induction. Notice that we \emph{cannot} use the same argument as Theorem~\ref{thrm:cont2opt0} since $s_a\rightarrow s_b$ does not necessarily imply $\T{i+1}(s_a)\rightarrow_o^{(i+1)}\T{i+1}(s_b)$.
  
  By Theorem~\ref{thrm:cont2opt0}, we know that \eqref{eq:setinclusionopt}
holds when $i = 1$. For the transition system $TS(\Sigma)$, suppose that the
controller $\gamma = (\gamma_1,\gamma_2,\dots)$ is winning for $W(TS(\Sigma),\varphi)$. For any $s[0]\in W(TS(\Sigma),\varphi)$ and environmental actions $e[0]e[1]\dots$, we create a controlled execution using $\gamma$: $\sigma = (s[0],e[0])(s[1],e[1])\dots$, which is winning.

  Let us define a hitting time $\tau$ as
  \begin{displaymath}
    \tau = \inf\{t\in\mathbb N_0:\T{i-1}(s[t])\in \mathcal W^{(i-1)}\}.
  \end{displaymath}
  In other words, $\tau$ is the first time that $\T{i-1}(s[t])$ enters the winning set $\mathcal W^{(i-1)}$. We further assume that the infimum over an empty set is $\infty$.
    
  For the FTS $\Dp{i-1}$, suppose that the controller $\gamma_p = (\gamma_{p,1},\dots)$ is winning for $W(\Dp{i-1},\varphi) = \mathcal W^{(i-1)}$. If $\tau < \infty$, we define $\varsigma_p[0] = \T{i-1}(s[\tau])$ and $e_p[t] = e[t+\tau]$. Now we create a controlled execution using $\gamma_p$ with environmental actions $e_p[0]e_p[+1]\dots$: $\sigma_p = (\varsigma_p[0],e_p[0])(\varsigma_p[1],e_p[1])\dots$, which is also winning.

  We now construct a controller $\gamma_o= (\gamma_{o,1},\dots)$ of the FTS $\Do{i}$, such that it is winning at $\varsigma[0] = \T{i}(s[0])$. The construction can by divided into two steps:
  \begin{enumerate}
    \item If $t\leq \tau$, then $\gamma_o$ follows the winning controller $\gamma$ of the FTS $TS(\Sigma)$, i.e.,
      \begin{align*}
	\gamma_{o,t}(\varsigma[0],&e[0],\dots,e[t-1]) \\
	&= \T{i}(\gamma_t(s[0],e[0],\dots,e[t-1])).
      \end{align*}
    \item If $t > \tau$, we switch to the winning controller $\gamma_p$ of the FTS $\Dp{i-1}$, i.e.,
      \begin{align*}
	\gamma_{o,t}&(\varsigma[0],e[0],\dots,e[t-1])\\
	&=(\gamma_{p,t-\tau}(\varsigma_p[0],e_p[0],\dots,e_p[t-\tau-1]),1).
      \end{align*}
  \end{enumerate}

  Now we prove that $\gamma_o$ is winning at $\varsigma[0]$. Define the controlled execution using $\gamma_o$ on the FTS $\Do{i}$ to be
  \begin{displaymath}
   \sigma_o = (\varsigma_o[0],e[0])(\varsigma_o[1],e[1])\dots \; .
  \end{displaymath}
  We need to prove that $\sigma_o$ satisfies the specification and $\gamma_o$ is consistent. The proof is divided into two cases depending on whether $\tau =\infty$ or $\tau < \infty$.

  \emph{ Case 1: $\tau = \infty$}

  By the definition of $\gamma_o$, we know that
  \begin{displaymath}
    \varsigma_o[t] = \T{i}(s[t]).
  \end{displaymath}
  Since $\sigma$ is winning, we only need to check the consistency of $\gamma_o$, i.e., whether the transition from $(\varsigma_o[t],e[t])$ to $(\varsigma_o[t+1],e[t+1])$ exists in $\Do{i}$. By Lemma~\ref{lemma:setinvariant}, we know that
  \begin{displaymath}
    s[t]\in W(TS(\Sigma),\varphi).
  \end{displaymath}
  And hence, by the induction assumption,  
  \begin{displaymath}
    \T{i-1}(s[t])\in \mathcal M^{(i-1)}\bigcup \mathcal W^{(i-1)}.
  \end{displaymath}
  By the fact that $\tau = \infty$,
  \begin{displaymath}
    \T{i-1}(s[t])\in \mathcal M^{(i-1)}.
  \end{displaymath}
  As a result, there exists an $j_t \in\{1,\dots,m\}$, such that $\varsigma_o[t]$ is the $j_t$th child of $\T{i-1}(s[t])$, i.e.,
  \begin{displaymath}
   \varsigma_o[t] = (\T{i-1}(s[t]),j_t).
  \end{displaymath}
  Furthermore, since there exists an $u[t]$, such that $f(s[t],u[t]) = s[t+1]$, we know that 
  \begin{displaymath}
    \mathcal R_o(\invT{i}(\varsigma_o[t]),\invT{i}(\varsigma_o[t+1]))=1,
  \end{displaymath}
  Hence, the transition from $(\varsigma_o[t],e[t])$ to
$(\varsigma_o[t+1],e[t+1])$ is an \emph{MM-transition} and it exists in
$\Do{i}$. And thus, $\gamma_o$ is consistent. 

  \emph{ Case 2: $\tau < \infty$}

  By the construction of $\gamma_o$, $\sigma_o$ satisfies
  \begin{displaymath}
    \varsigma_o[t] = \begin{cases}
      \T{i}(s[t])&\text{if }t\leq \tau,\\
      (\varsigma_p[t-\tau],1)&\text{if }t >  \tau.\\
    \end{cases}
  \end{displaymath}
  By Lemma~\ref{lemma:timeinvariant} and the fact that both $\sigma$ and $\sigma_p$ satisfy $\varphi$, we only need to check the consistency of $\gamma_o$, i.e., whether the transition from $(\varsigma_o[t],e[t])$ to $(\varsigma_o[t+1],e[t+1])$ exists in $\Do{i}$. This can be done in three steps:
  \begin{enumerate}
    \item $t<\tau-1$:
      
      By the same argument as for the case where $\tau = \infty$, we know that the transition from $(\varsigma_o[t],e[t])$ to $(\varsigma_o[t+1],e[t+1])$ is an \emph{MM-transition} and it exists in $\Do{i}$.
    \item $t=\tau-1$: 

      By the definition of $\tau$, we know that
      \begin{displaymath}
	\T{i-1}(s[\tau-1])\in \mathcal M^{(i-1)},\, \T{i-1}(s[\tau])\in \mathcal W^{(i-1)}.
      \end{displaymath}
      Hence, the transition from $(\varsigma_o[\tau-1],e[\tau-1])$ to $(\varsigma_o[\tau],e[\tau])$ is an \emph{MW-transition} and it exists in $\Do{i}$.
    \item $t>\tau-1$:

      By Lemma~\ref{lemma:setinvariant}, we know that
      \begin{displaymath}
	\varsigma_p[t] \in W(\Dp{i-1},\varphi) = \mathcal W^{(i-1)}.	
      \end{displaymath}
      Hence, the transition from $(\varsigma_o[t],e[t])$ to $(\varsigma_o[t+1], e[t+1])$ is a \emph{WW-transition} and it exists in $\Do{i}$. 
  \end{enumerate}
Therefore, $\gamma_o$ is consistent and we can conclude the proof.
\end{proof}
}

\end{document}